\documentclass{article}
\usepackage[left=2.5cm,right=2.5cm,top=2.5cm,bottom=2.5cm]{geometry}
\usepackage[english]{babel}
\usepackage[T1]{fontenc}
\usepackage{graphicx}
\usepackage{amsmath, amsthm, amssymb, mathrsfs}
\usepackage{epstopdf}
\usepackage{caption}
\usepackage{subcaption}
\usepackage{float}
\usepackage{appendix}
\usepackage{tikz}
\usetikzlibrary{shapes,positioning}
\usepackage{braket}
\usepackage{anyfontsize}
\usepackage{mathptmx}
\usepackage[11pt]{moresize}
\usepackage{enumerate}
\usepackage{dcolumn}
\usepackage{bm}
\usepackage{tabularx}
\usepackage{array}
\usepackage{lipsum}
\usepackage{cancel}
\usepackage{enumitem}
\usepackage{soul}
\usepackage{authblk}
\usepackage{pgfplots}
\usepackage{cite}
\usepackage[percent]{overpic}

\def\cH{{\mathcal H}}
\def\cN{{\mathcal N}}

\newcommand{\be}{\begin{equation}}
\newcommand{\ee}{\end{equation}}

\newcommand{\ba}{\begin{align}}
\newcommand{\ea}{\end{align}}

\newcommand{\bea}{\begin{eqnarray}}
\newcommand{\eea}{\end{eqnarray}}
\newcommand{\bt}{\begin{theorem}}
\newcommand{\et}{\end{theorem}}
\newcommand{\bp}{\begin{proposition}}
\newcommand{\ep}{\end{proposition}}
\newcommand{\bd}{\begin{definition}}
\newcommand{\ed}{\end{definition}}

\newtheorem{theorem}{Theorem}
\newtheorem{proposition}{Proposition}
\newtheorem{definition}{Definition}

\newtheorem{remark}{Remark}

\newtheorem{lemma}[theorem]{Lemma}
\newtheorem{corollary}[theorem]{Corollary}
\begin{document}

\title{Discrimination of bosonic dephasing quantum channels}

\author[1]{Samad Khabbazi Oskouei%
	\thanks{\texttt{kh.oskouei@iauvaramin.ac.ir}}}
\author[2]{Laleh Memarzadeh%
	\thanks{\texttt{memarzadeh@sharif.edu}}}
\author[3,4]{Stefano Mancini%
	\thanks{\texttt{stefano.mancini@unicam.it}}}
\author[3,4]{Milajiguli Rexiti%
	\thanks{\texttt{milajiguli.milajiguli@unicam.it}; Corresponding author}}

\affil[1]{Department of Mathematics, Varamin-Pishva Branch Islamic Azad University, 33817-7489 Varamin, Iran}
\affil[2]{Department of Physics, Sharif University of Technology,  11155-9161 Tehran, Iran}
\affil[3]{School of Science and Technology, University of Camerino, Via Madonna delle Carceri 9, I-62032 Camerino, Italy}
\affil[4]{INFN–Sezione di Perugia, Via A. Pascoli, I-06123 Perugia, Italy}

\maketitle	

\begin{abstract}
We study the possibility of discriminating between two bosonic dephasing quantum channels.
We show that unambiguous discrimination is not realizable. We then consider discrimination with nonzero error probability and minimize this latter in the absence of input constraints. In the presence of an input energy constraint, we derive an upper bound on the error probability. 
Finally, we extend these results from single-shot to multi-shot discrimination, envisaging the asymptotic behavior.
\end{abstract}


\section{Introduction}

Since any physical process can be formally described by a quantum channel, i.e. a linear completely positive and trace preserving map on the space of states, the discrimination between two or more quantum channels became a relevant task in quantum information theory \cite{GLN05, S05, WY06, H09, MrSm, DFY09}. 
In infinite-dimensional Hilbert space, the attention has been confined to the discrimination of Gaussian channels because these represent cases where the problem can be treated analytically.
Actually, we can say that the characterization of continuous variable maps beyond the Gaussian perimeter has just started.
It involves the bosonic dephasing channel for which some results about quantum capacities have been obtained \cite{AMM20, Dehdashti2022QuantumCO, LW23,AMM23}. {These progresses also motivate the characterization of the capabilities in distinguishing between two bosonic dephasing quantum channels. A work along this direction has been recently reported \cite{PRXQuantum}. There, the results are derived assuming unconstrained input states. Here, we shall consider the possibility of having input energy constraint, as this becomes important from a practical perspective.
 }

We show that unambiguous discrimination is not accomplishable with bosonic dephasing quantum channels. Then, we find the optimal error probability minimized over all possible inputs, {providing an alternative derivation with respect to Ref.\cite{PRXQuantum}.}
 When the input energy constraint is present, we derive an upper bound on the error probability. Finally, we extend 
these results from single-shot to multi-shot discrimination, envisaging the asymptotic behavior of the performance.


\section{The bosonic dephasing quantum channel}

Let ${\cal B}(\cH)$ be the set of bounded operators on an infinite dimensional separable Hilbert space $\cH$. 
Considering the subset ${\cal T}(\cH)\subset {\cal B}(\cH)$ of trace-class operators, any quantum channel 
 results a linear and completely positive trace preserving map $\mathcal{N}:{\cal T}(\cH)\to{\cal T}(\cH)$.

The action of the bosonic dephasing quantum channel ${\cN}_{\gamma}: \mathcal{T}(\cH)\rightarrow \mathcal{T}(\cH)$, with dephasing rate $\gamma\geq 0$, can be  defined as follows \cite{AMM20}
\begin{equation}\label{defNg}
  \cN_{\gamma}(\rho)=\int_{-\infty}^{\infty} d\theta p'_{\gamma}(\theta) e^{-ia^\dagger a\theta }\rho e^{ia^\dagger a\theta } ,
\end{equation}
where $p'_\gamma$ is the normal distribution
\begin{equation}
p'_\gamma(\theta)=\frac{1}{\sqrt{2\pi\gamma}}e^{-\frac{\theta^2}{2\gamma}}.
\end{equation}
Here $a^\dagger, a$ are the bosonic ladder operators on $\cH$. 

Notice that Eq.\eqref{defNg} can be recast in the form presented in Ref.\cite{LW23} 
\begin{equation}\label{Ng}
  \cN_{\gamma}(\rho)=\int_{-\pi}^{\pi} d\theta p_{\gamma}(\theta) e^{-ia^\dagger a\theta }\rho e^{ia^\dagger a\theta },
\end{equation}
with $p_\gamma$ the wrapped normal distribution on $[-\pi,\pi]$, 
i.e.\footnote{ This can also be regarded as the Weil transform 
of the function $e^{-\theta^2/(2\gamma)}/\sqrt{2\pi\gamma} $ and expressed in terms of Jacobi theta function
as $\Theta(i\theta/\gamma;2\pi i/\gamma) e^{-\theta^2/(2\gamma)}/\sqrt{2\pi\gamma}$. }
\bea
p_\gamma(\theta)=\frac{1}{\sqrt{2\pi\gamma}}\sum_{k=-\infty}^{+\infty}e^{-\frac{1}{2\gamma}(\theta+2\pi k)^2}.
\eea
In fact, from Eq.\eqref{defNg}, we can write
\begin{equation}
\cN_{\gamma}(\rho)=\sum_{k\in \mathbb{Z}}\int_{2\pi k-\pi}^{2\pi k+\pi} d\theta p'_{\gamma}(\theta) e^{-ia^\dagger a\theta }\rho e^{ia^\dagger a\theta }.
\end{equation}
Defining a new variable $\alpha=\theta-2\pi k$, we find
\begin{equation}
\cN_{\gamma}(\rho)=\sum_{k\in \mathbb{Z}}\int_{-\pi}^{\pi} d\alpha p'_{\gamma}(\alpha+2\pi k) 
e^{-ia^\dagger a \alpha{ -}2\pi k ia^\dagger a }\rho e^{ia^\dagger a \alpha+2\pi k ia^\dagger a }.
\end{equation}
However, 
{ 
expressing the input state in the Fock basis
\be\label{rhoFock}
\rho=\sum_{m,n=0}^{\infty} \rho_{m,n}\ket m \bra n,
\ee
and noticing that 
$ e^{2\pi k ia^\dagger a }\vert m\rangle
 =\vert m\rangle$,}
we can write
\begin{equation}
\cN_{\gamma}(\rho)=\sum_{k\in \mathbb{Z}}\int_{-\pi}^{\pi} d\alpha p'_{\gamma}(\alpha+2\pi k) e^{-ia^\dagger a \alpha}\rho e^{ia^\dagger a \alpha }.
\end{equation}
{ Making explicit use of Eq.\eqref{rhoFock}}
we can also obtain
\be\label{cn}
\cN_\gamma \left(\rho\right)={\sum_{m,n=0}^{\infty} \int_{-\pi}^{\pi} d\theta p_{\gamma}(\theta){ e^{-i\theta(m-n) }}\rho_{m,n}\ket m \bra n  }=\sum_{m,n=0}^{\infty} e^{-\frac{1}{2}\gamma \left(m-n\right)^2} \rho_{m,n}\ket m \bra n.
\ee


\begin{definition}\label{defHprod}
  Let us assume that $A$ and $B$ are two $n\times n$ matrices. The Hadamard product $A\circ B$ is an $n\times n$ matrix whose entries are $(A\circ B)_{ij}=A_{ij} B_{ij}$, $i, j=1, \cdots, n$.
\end{definition}

For more information about Hadamard matrices, see Ref.\cite{HJ91}.

\begin{definition}\label{defTmatrix}
  A Toeplitz matrix of size $M\in\mathbb{N}$ is defined as follows
  \begin{equation}\label{Tmatrix}
  T_M(t)=\left(
  \begin{array}{ccccc}
    t_0 & t_{-1} & t_{-2} & \cdots & t_{-M+1} \\
       t_1 & t_0 & t_{-1} & \cdots & t_{-M+2} \\
          t_2 & t_1 & t_0 & \cdots & t_{-M+3} \\
               \vdots & \vdots & \vdots & \ddots & \vdots \\
                    t_{M-1} & t_{M-2} & t_{M-3} & \cdots & t_0 \\
                       \end{array}
                          \right),
   \end{equation}
   where $t_{-M+1}, \cdots, t_{-1}, t_0, t_1,\cdots, t_{M-1}$ are complex numbers. In other words, the entries of a Toeplitz matrix are written as 
   $[T_M(t)]_{mn}=t_{m-n}$, for $m, n=0, 1, \cdots, M-1$. For our purposes, we will consider symmetric infinite dimensional Toeplitz matrices, where $t_0, t_1, \cdots$ are Fourier coefficients of an absolutely integrable function $t\in L^{(1)}[-\pi,\pi]$. Specifically,
\begin{equation}
  t_n=\frac{1}{2\pi} \int_{-\pi}^{\pi} {\rm d\theta}\,\, t(\theta) e^{-i\theta n}.
\end{equation}
\end{definition}

Using Definitions \ref{defHprod}, \ref{defTmatrix}, the quantum dephasing channel \eqref{cn} can be written as follows
\begin{equation}\label{re:14}
  \cN_\gamma \left(\rho\right)=\rho\circ T_\infty(p_{\gamma}),
\end{equation}
where $ T_\infty(p_{\gamma})$ is the infinite dimensional Toeplitz matrix with entries
\bea\label{re:15}
\left[T_\infty(p_\gamma)\right]_{m,n}\equiv t_{m-n} =\frac{1}{2\pi}\int_{-\pi}^{\pi}  p_\gamma(\theta)e^{i\theta(m-n)} {\rm d}\theta.
\eea


\section{Preliminaries on the discrimination of dephasing channels}

Given  two bosonic dephasing channels ${\cN}_{\gamma_1}$ and ${\cN}_{\gamma_2}$, { with $\gamma_1\neq \gamma_2$, 
occurring with probabilities $q_1$ and $q_2$ respectively ($q_1+q_2=1$)}, we aim to discriminate between them. This discrimination can eventually be subjected to an energy constraint on the channels' input, that is ${\rm Tr}\left( \rho H\right)\leq E$, where $H$ is the energy observable
(whose eigenstates $H|n\rangle=n|n\rangle, \, n=0,1, \ldots$ form the Fock basis of $\cH$) and $E\geq 0$.

In principle, we can proceed with two different strategies. One is to perform unambiguous discrimination, where there are no errors, but an inconclusive outcome is possible. The goal here is to minimize the probability of such an event. The other strategy is to allow for wrong outcomes (channel ${\cN}_{\gamma_1}$  misidentified as ${\cN}_{\gamma_2}$ and vice-versa) and minimize the error probability.

\begin{remark}
Since ${\rm Tr}$ is a continuous linear functional, for every trace class operator, we can exchange the
integral (or sum) with ${\rm Tr}$. This fact will be often used in the following.
\end{remark}


\subsection{Unambiguous Discrimination}

The goal is to find a POVM  $\{\Pi_0,\Pi_1,\Pi_2\}$, where elements correspond to inconclusive outcome, channel ${\cN}_{\gamma_1}$ outcome, and channel ${\cN}_{\gamma_2}$ outcome respectively. In order not to have errors, it must be 
\begin{equation}\label{ortcond}
{\rm Tr}\left(\Pi_1 {\cN}_{\gamma_2}(\rho)\right)={\rm Tr}\left(\Pi_2 {\cN}_{\gamma_1}(\rho)\right)=0.
\end{equation}

\begin{theorem}
There does not exist a POVM that implements unambiguous discrimination of the dephasing channel.
\end{theorem}

\begin{proof}
According to  Eq.\eqref{re:15}, for any density matrix and positive semi-definite operator $\Pi\leq I$,
\begin{equation}
  {\rm Tr}(\Pi T_\infty(p_{\gamma})\circ \rho)={\rm Tr} \left( \Pi \sum_{m,n=0}^{\infty} \int_{-\pi}^{\pi} {\rm d}\theta p_{\gamma}(\theta)e^{i\theta(n-m)}  \rho_{m,n}\ket m \bra n  \right).
\end{equation}
Defining $U(\theta)$ as a diagonal matrix in the Fock basis with entries  $e^{ik\theta}$, $k=0,1,2,\cdots$, we have
\begin{align}
{\rm Tr}(\Pi T_\infty(p_{\gamma})\circ \rho) &= {\rm Tr}\left(  \int_{-\pi}^{\pi} {\rm d}\theta p_{\gamma}(\theta) \sqrt{\Pi} U(\theta)^\dagger \rho  U(\theta)\sqrt{\Pi} \right)\\
&= \int_{-\pi}^{\pi} {\rm d}\theta  p_{\gamma}(\theta){\rm Tr}\left(\sqrt{\Pi} U(\theta)^\dagger\rho U(\theta)\sqrt{\Pi}\right)\\
&=   \int_{-\pi}^{\pi} {\rm d}\theta  p_{\gamma}(\theta){\rm Tr}\left(U(\theta) \Pi U(\theta)^\dagger\rho\right)\geq 0 \label{re:12}.
\end{align}

The last inequality comes from the fact that $p_\gamma(\theta)\geq 0$, and $\rho$ is a density matrix.
If we now set ${\rm Tr}(\Pi T_\infty(p_{\gamma})\circ \rho)=0$, then we should have ${\rm Tr}\left(U(\theta) \Pi U(\theta)^\dagger\rho\right)=0$.

According to \eqref{ortcond}, let us assume that ${\rm Tr}(\Pi_1 T_\infty(p_{\gamma_2})\circ \rho)={\rm Tr}(\Pi_2 T_\infty(p_{\gamma_1})\circ \rho)=0$ for given two positive semi-definite operators $\Pi_1, \Pi_2\neq 0$ and $\Pi_1, \Pi_2\leq I$. 
Then, by means of relation \eqref{re:12}, we find for every $\theta$ that
\begin{equation}
  {\rm Tr}\left(U(\theta) \Pi_1 U(\theta)^\dagger\rho\right)={\rm Tr}\left(U(\theta) \Pi_2 U(\theta)^\dagger\rho\right)=0.
\end{equation}
Therefore,
\begin{multline}
  {\rm Tr}(\Pi_1 T_\infty(p_{\gamma_1})\circ \rho)+{\rm Tr}(\Pi_2 T_\infty(p_{\gamma_2})\circ \rho)=\\
  \int_{-\pi}^{\pi} {\rm d}\theta  \left[
  p_{\gamma_1}(\theta)  {\rm Tr}\left( U(\theta)\Pi_1 U(\theta)^\dagger\rho\right)    +p_{\gamma_2}(\theta){\rm Tr}\left( U(\theta)\Pi_2 U(\theta)^\dagger\rho\right)\right]=0.
\end{multline}
This simply means that the probability of unambiguously identify ${\cN}_{\gamma_1}$ and ${\cN}_{\gamma_2}$ vanishes.
\end{proof}


\subsection{Minimum error probability}

If we allow errors in the discrimination of channels output states ${\cN}_{\gamma_1}(\rho)$ and ${\cN}_{\gamma_2}(\rho)$, 
it is known (see e.g. \cite{MW20}) that the minimum error probability reads
{
\begin{equation}\label{perr}
  p_{err}(\cN_{\gamma_1}(\rho),\cN_{\gamma_2}(\rho))=\frac{1}{2}\left( 1-\|{q_1}\cN_{\gamma_1}(\rho)-{q_2}\cN_{\gamma_2}(\rho)\|_{1}\right),
\end{equation}
}
 where
\begin{equation}\label{eq:TNorm}
\| A\|_1=\rm{Tr}\left(\sqrt{ A^{\dag}A}\right).
\end{equation}
For optimal channel discrimination, the probability \eqref{perr} should be minimized over all possible inputs.
Let us define  
\begin{equation}
\Delta{(\rho)}:={ q_1}\cN_{\gamma_1}(\rho)-{ q_2}\cN_{\gamma_2}(\rho). 
\end{equation}
Then, we have the following optimization problem
\bea
\max_{\rho\in {{\mathcal{T}(\cH)}}}\|\Delta(\rho)\|_{1}
=\max_{\rho\in {{\mathcal{T}(\cH)}}}\|{ q_1}\cN_{\gamma_1}(\rho)-{ q_2}\cN_{\gamma_2}(\rho)\|_{1}.
\eea
On the other hand, we know that the trace norm satisfies the convexity property, hence
\begin{equation}\label{re:31}
\max_{\rho\in {{\mathcal{T}(\cH)}}}\|\Delta(\rho)\|_{1}=
  \max_{ \rho\in \mathcal{P}(\cH)}\|\Delta({ \rho})\|_{1},
\end{equation}
where $ \mathcal{P}(\cH)$ denotes the set of pure states on $\cH$.

When we have an input energy constraint, the problem becomes
\bea
\max_{\rho\in {{\mathcal{P}(\cH)}}}\|\Delta(\rho)\|_{1}, \quad
{\text {s.t.}} \quad{\rm Tr}\left( \rho H\right)\leq E.
\label{re:dis}
\eea



\section{Single-shot discrimination}

In this section, we analyze the minimum error probability when discriminating two bosonic dephasing quantum channels 
in a single-shot scenario. 

\begin{theorem}\label{thm:theorem2}
Without an input energy constrain, it is
  \begin{equation}
 \max_{\rho\in \mathcal{P}(\cH)}\|\Delta(\rho)\|_{1}=\int_{-\pi}^{\pi} {\rm d}\theta \left\vert{ q_1} p_{\gamma_1}(\theta)-{ q_2}p_{\gamma_2}(\theta)\right\vert.
\end{equation}
\end{theorem}

\begin{proof}
Using Relation \eqref{cn}, we have
\begin{align}\label{eqUB}
\max_{\rho\in \mathcal{P}(\cH)}\|\Delta(\rho)\|_{1}& =  \left\Vert  \sum_{m,n=0}^{\infty} \int_{-\pi}^{\pi} {\rm d}\theta \left({ q_1}p_{\gamma_1}(\theta)-{ q_2}p_{\gamma_2}(\theta)\right)e^{i\theta(n-m)}  \rho_{m,n}\ket m \bra n  \right\Vert_1\\
&= \left\Vert  \int_{-\pi}^{\pi} {\rm d}\theta \left({ q_1}p_{\gamma_1}(\theta)-{ q_2}p_{\gamma_2}(\theta)\right)\sum_{m,n=0}^{\infty}   e^{i\theta(n-m)}  \rho_{m,n}\ket m \bra n  \right\Vert_1.
\end{align}
Using $U(\theta)$ defined as a diagonal matrix in the Fock basis with entries  $e^{ik\theta}$, $k=0,1,2,\cdots$, we have
\begin{align}
\max_{\rho\in \mathcal{P}(\cH)}\|\Delta(\rho)\|_{1} &= \left\Vert  \int_{-\pi}^{\pi} {\rm d}\theta \left({ q_1}p_{\gamma_1}(\theta)-{ q_2}p_{\gamma_2}(\theta)\right) U(\theta)^\dagger\rho U(\theta) \right\Vert_1\\
&\leq \int_{-\pi}^{\pi} {\rm d}\theta \left\vert\left({ q_1}p_{\gamma_1}(\theta)-{ q_2}p_{\gamma_2}(\theta)\right)\right\vert \left\Vert U(\theta)^\dagger\rho U(\theta)\right\Vert_1\\
&=   \int_{-\pi}^{\pi} {\rm d}\theta \left\vert { q_1}p_{\gamma_1}(\theta)-{ q_2}p_{\gamma_2}(\theta)\right\vert\label{re:34}.
\end{align}
The last equality arises from the invariance of Shatten's norms under unitary transformations.

\bigskip

To prove the reverse inequality, we will resort to the following Theorem (see \cite{Sz20, GS84}).
\begin{theorem}\label{thm:sezgo}
Let $t:[-\pi, \pi]\to \mathbb{R}$ be an absolutely integrable function with $\inf_{x\in[-\pi,\pi]} t(x)$, $\sup_{x\in[-\pi,\pi]} t(x)$ finite numbers.    
If $T_M(t)$ is a $M\times M$ Toeplitz matrix and $\lambda_j(T_M(t))$ denotes the $j^{th}$ eigenvalue of $T_M(t)$, then 
	\begin{equation}
		\lim_{M\to\infty} \frac{1}{M}\sum_{j=1}^{M} F(\lambda_j(T_M(t)))=\frac{1}{2\pi}\int_{-\pi}^{\pi} {\rm d}\theta F(t(\theta)),
	\end{equation}
	where $F:[\inf_{x\in [-\pi,\pi]} t(x), \sup_{x\in [-\pi, \pi]} t(x)]\to  \mathbb{R}$ is any continuous function.
\end{theorem}

For a given positive $M\times M$ Toeplitz matrix
\begin{equation}\label{re:35}
  T_M=\left(
           \begin{array}{ccc}
             1 & \cdots & 1 \\
             \vdots & \cdots & \vdots \\
             1 & \cdots & 1 \\
           \end{array}
         \right),
\end{equation}
we consider the density matrix $\rho_M=\frac{1}{M}T_M$. We can extend $\rho_M$ to a density matrix on infinite dimensional Hilbert space $\cH$ with a canonical injection ($\cH_M\subset\cH_{M+1}\subset \ldots\subset\cH$, where $\cH_M$ is the $M$-dimensional Hilbert space spanned by the first $M$ vectors of the Fock basis of $\cH$), which implies
\begin{align}
\|\Delta(\rho_M)\|_{1}&=\Vert T_\infty({ q_1}p_{\gamma_1}-{ q_2}p_{\gamma_2})\circ \rho_M\Vert_1\\
&= \frac{1}{M}\Vert P_M T_\infty ({ q_1}p_{\gamma_1}-{ q_2}p_{\gamma_2})P_M\Vert_1,
\end{align}
where $P_M=\sum_{k=0}^{M-1} \vert k\rangle\langle k\vert$ is the projector onto $\cH_M$. Now, using Theorem \ref{thm:sezgo} with $F(x)=\vert x\vert$, we have 
  \begin{align}
  \max_{\rho\in \mathcal{T}(\cH)}\|\Delta(\rho)\|_{1} &\geq \lim_{M\to \infty}\|\Delta(\rho_M)\|_{1}\\
  &=2\pi \lim_{M\to \infty} \left\Vert \frac{1}{2\pi M} \sum_{m,n=0}^{M-1} \left({ q_1}e^{-\frac{1}{2}\gamma_1 \left(m-n\right)^2}-{ q_2}e^{-\frac{1}{2}\gamma_2 \left(m-n\right)^2}\right) \ket m \bra n  \right\Vert_1\\
  &=\int_{-\pi}^{\pi} {\rm d}\theta \left\vert { q_1}p_{\gamma_1}(\theta)-{ q_2}p_{\gamma_2}(\theta)\right\vert.
\end{align}
We emphasize that  the Fourier coefficients in the first equality are $\frac{1}{2\pi}e^{-\frac{1}{2}\gamma \left(m-n\right)^2}$.
\end{proof}

Notice that the result of Theorem \ref{thm:theorem2} was already established in Ref.\cite{PRXQuantum} (see Eq.(8) there);
 however, here it is derived with a different approach.


Unfortunately, in the presence of input energy constraint, we cannot achieve equality \eqref{re:34} likewise in the proof of Theorem \ref{thm:theorem2}.
In fact, we should have  
${\rm Tr}(\rho H)\leq E$. However, setting $\rho_M=T_M/\Vert T_M\Vert_1$, we get  from Eq.\eqref{Tmatrix}
\begin{equation}
{\rm Tr}\left(\sum_{k, i, j=0}^{ M-1} k  \vert i\rangle\langle j\vert\vert k\rangle\langle k\vert\right) \leq E \Vert T_M\Vert_1 .
\end{equation}
It results
\begin{equation}
  \frac{M(M+1)}{2} \leq M E ,
\end{equation}
which is impossible to satisfy for large $M$.

With the following theorem, we will provide a simple upper bound to the error probability in presence of input energy constraint.
{ 
 
\begin{theorem}\label{thm3}
For a given energy $E>0$, we have the following:
\begin{itemize}
  \item If $E\geq 1/2$, then
\begin{equation}\label{re:45}
     \max_{ \rho\in \mathcal{P}(\cH),\, {\rm Tr}(\rho H)\leq E}\|\Delta({ \rho})\|_{1}\geq
    \max_{m_*\in\{1,m_{ext},\lfloor 2E\rfloor\}}
    \left\vert  q_1 e^{-\frac{m_*^2}{2}\gamma_1}- q_2 e^{-\frac{m_*^2}{2}\gamma_2}\right\vert,
\end{equation}
with
\begin{equation}    
    m_{ext}=\left[ \sqrt{2 \frac{\ln(q_1\gamma_1)-\ln(q_2\gamma_2)}{\gamma_1-\gamma_2}} \right],
        \quad \text{and} \quad
 1\leq  m_{ext} \leq \lfloor 2E\rfloor,
\end{equation}
being $\lfloor\bullet\rfloor$ the biggest integer that is smaller than $\bullet$, and $[\bullet]$ the closest integer to $\bullet$.
  \item If $E < 1/2$, then
\begin{equation}\label{re:46}
     \max_{ \rho\in \mathcal{P}(\cH),\, {\rm Tr}(\rho H)\leq E}\|\Delta({ \rho})\|_{1}\geq
     \max_{m_*\in\{1,m_{ext}\}}
    2\sqrt{\frac{E}{m_*} \left(1-\frac{E}{m_*}\right)} \,  
    \left\vert  q_1 e^{-\frac{m_*^2}{2}\gamma_1}- q_2 e^{-\frac{m_*^2}{2}\gamma_2}\right\vert,
\end{equation}
\end{itemize}
\end{theorem}

\begin{proof}
Consider the following pure qubit state embedded in $\mathcal{P}(\cH)$
\begin{equation}
  \rho=r_0 \vert 0\rangle\langle 0\vert +\sqrt{r_0r_m}\vert 0\rangle\langle m\vert+\sqrt{r_0r_m}\vert m\rangle\langle 0\vert
  +r_m\vert m\rangle\langle m\vert,
\end{equation}
where $r_0$ and $r_m$ are non-negative real numbers satisfying $r_0+r_m = 1$.

Using relation \eqref{cn}, for given $\gamma_1$ and $\gamma_2$, we have
\begin{align}
  \left({q_1}T_\infty(p_{\gamma_1})-{q_2}T_\infty(p_{\gamma_2})\right)\circ \rho
  &=\sqrt{r_0r_m}\left( q_1 e^{-\frac{m^2}{2}\gamma_1}- q_2 e^{-\frac{m^2}{2}\gamma_2 }\right)\vert 0\rangle\langle m\vert
  +\sqrt{r_0r_m}\left( q_1 e^{-\frac{m^2}{2}\gamma_1 }- q_2 e^{-\frac{m^2}{2}\gamma_2 }\right)\vert m\rangle\langle 0\vert.
\end{align}
It follows that
\begin{equation}\label{eq:r1r2}
\left\Vert {q_1}T_\infty(p_{\gamma_1})-{q_2}T_\infty(p_{\gamma_2})\circ \rho\right\Vert_1=2\left\vert q_1 e^{-\frac{m^2}{2}\gamma_1}- q_2 e^{-\frac{m^2}{2}\gamma_2}   \right\vert\sqrt{r_0r_m}.
\end{equation}
Now, from the energy constraint, we have 
\begin{equation}\label{eq:constr}
{\rm Tr}(\rho H)\leq E \; \Longrightarrow \; m r_m\leq E. 
\end{equation}

For the case $E\geq 1/2$, we can take the maximum value $\frac{1}{2}$ of the term $\sqrt{r_0r_m}$ in Eq.\eqref{eq:r1r2}, and then, we are left with the 
expression $\left\vert q_1 e^{-\frac{m^2}{2}\gamma_1}- q_2 e^{-\frac{m^2}{2}\gamma_2} \right\vert$, which we can maximize over $m$ obtaining the maximum either at the extremum
\begin{equation}    
    m_{ext}=\left[ \sqrt{2 \frac{\ln(q_1\gamma_1)-\ln(q_2\gamma_2)}{\gamma_1-\gamma_2}} \right],
\end{equation} 
or at the border points $m=1$, $m=\lfloor 2E\rfloor$, given that  Eq.\eqref{eq:constr} now implies 
\begin{equation}
 1\leq  m \leq \lfloor 2E\rfloor. 
\end{equation}

Instead, for $E<1/2$ we cannot take $r_0=r_m=1/2$. However, since $r_0r_m=r_m(1-r_m)$ is monotonically
increasing in the interval $[0, 1/2 ]$,  we take $r_m = E/m_*$ (saturating the bound \eqref{eq:constr}).

Then, Eq.\eqref{eq:r1r2} becomes
\begin{equation}
\left\Vert {q_1}T_\infty(p_{\gamma_1})-{q_2}T_\infty(p_{\gamma_2})\circ \rho\right\Vert_1=2\left\vert q_1 e^{-\frac{m_*^2}{2}\gamma_1}- q_2 e^{-\frac{m_*^2}{2}\gamma_2}   \right\vert\sqrt{\frac{E}{m_*}\left(1-\frac{E}{m_*}\right)}.
\end{equation}

\end{proof}
}

{ The term $\sqrt{r_0r_m}$ in Eq.\eqref{eq:r1r2} represents the off diagonal elements of $\rho$ which cause the decoherence and facilitate the discrimination. This leads to the threshold $E=1/2$ in Theorem \ref{thm3}, which can then be understood as consequence of the fact that we restricted our attention to states in a two dimensional subspace. It is clear that the lower bound for $E\geq 1/2$ is not tight, as it might not depend on $E$. The following theorem provides a tighter lower bound in such a case.}

\begin{theorem}\label{thm:energyconstriant}
For a given energy $E\geq 1/2$, we have 
\begin{equation}\label{re:56}
  \max_{ \rho\in \mathcal{P}(\cH),\, {\rm Tr}(\rho H)\leq E}\|\Delta({ \rho})\|_{1}
   \geq \frac{1}{\lfloor 2E\rfloor+1} \left\Vert P_{\lfloor 2E\rfloor+1}\left ({ q_1}T_\infty(p_{\gamma_1})-{ q_2}T_\infty(p_{\gamma_2})\right)\right\Vert_1,
  \end{equation}
  where $P_{\lfloor 2E\rfloor+1}=\sum_{n=0}^{\lfloor 2E\rfloor} \vert n\rangle\langle n\vert$. 
  \end{theorem}

\begin{proof}
Let us  define a state
\begin{equation}
\vert \psi_M\rangle:=\frac{1}{\sqrt{M+1}}(\vert 0\rangle+\vert 1\rangle+\cdots+\vert M\rangle),
\end{equation}
where $M=\lfloor 2E\rfloor$. Since $M\leq 2E$, we get
\begin{equation}
  \langle \psi_M\vert H \vert \psi_M\rangle\leq E.
\end{equation}
This implies
\begin{align}
  \max_{ \rho\in \mathcal{P}(\cH),\, {\rm Tr}(\rho H)\leq E}\|\Delta({ \rho})\|_{1}
  &\geq  \left\Vert \left({ q_1}T_\infty(p_{\gamma_1})-{ q_2}T_\infty(p_{\gamma_2})\right)\circ \vert \psi_M\rangle\langle\psi_M\vert\right\Vert_1\\
 &{ =}\frac{1}{\lfloor 2E\rfloor+1} \left\Vert P_{\lfloor 2E\rfloor+1}\left ({ q_1}T_\infty(p_{\gamma_1})-{ q_2} T_\infty(p_{\gamma_2})\right)\right\Vert_1,
  \end{align}
  where $P_{\lfloor 2E\rfloor+1}=\sum_{n=0}^{\lfloor 2E\rfloor} \vert n\rangle\langle n\vert$. 
\end{proof}

For a given matrix $A$, we know that $\Vert A\Vert_F \leq \Vert A\Vert_1$, where $\Vert \cdot\Vert_F$ is the Frobenius norm defined as $\Vert A\Vert_F:= \sqrt{{\rm Tr}(A A^\dagger)}$. From this, we immediately arrive at the following Corollary of Theorem \ref{thm:energyconstriant}.

\begin{corollary}\label{thm:Frobeniusenergyconstriant}
For a given energy $E\geq 1/2$, we have 
\begin{align}
  \max_{ \rho\in \mathcal{P}(\cH),\, {\rm Tr}(\rho H)\leq E}\|\Delta({ \rho})\|_{1}
   &\geq \frac{1}{\lfloor 2E\rfloor+1} \sqrt{\sum_{m, n=0}^{\lfloor 2E\rfloor}\left\vert \langle m\vert \left({ q_1}T_\infty(p_{\gamma_1})-{ q_2}T_\infty(p_{\gamma_2})\right)\vert n\rangle \right\vert^2}\\
   &= \frac{1}{\lfloor 2E\rfloor+1} \sqrt{\sum_{k=1}^{\lfloor 2E\rfloor}  (\lfloor 2E\rfloor-k)\left\vert \langle 0\vert \left({ q_1}T_\infty(p_{\gamma_1})-{ q_2}T_\infty(p_{\gamma_2})\right)\vert k\rangle \right\vert^2}.
  \end{align}
\end{corollary}
  Using  Eq.\eqref{Tmatrix}, the equality comes from
  $$\langle m\vert \left({ q_1}T_\infty(p_{\gamma_1})-{ q_2}T_\infty(p_{\gamma_2})\right)\vert n\rangle=\langle 0\vert \left({ q_1}T_\infty(p_{\gamma_1})-{ q_2}T_\infty(p_{\gamma_2})\right)\vert k\rangle,$$
  where $k=|m-n|$.

\bigskip

Now, we are going to compare the lower bounds found in Theorem \ref{thm3}, Theorem \ref{thm:energyconstriant}, and Corollary \ref{thm:Frobeniusenergyconstriant}, together with the result of Theorem \ref{thm:theorem2}. Clearly, by virtue of relation \eqref{perr}, they will provide upper bounds for the error probability.
From Fig. \ref{fig1} we can see that as long as $E<1/2$, Theorem \ref{thm3}, 
provides the best bound. However, when $E\geq 1/2$, Theorem \ref{thm:energyconstriant} provides the best bound, 
which tends to saturate to the result of Theorem \ref{thm:theorem2} for $E\gg 1/2$.

\medskip

\begin{figure}[H]
\begin{center}
          \includegraphics[width=0.55\textwidth]{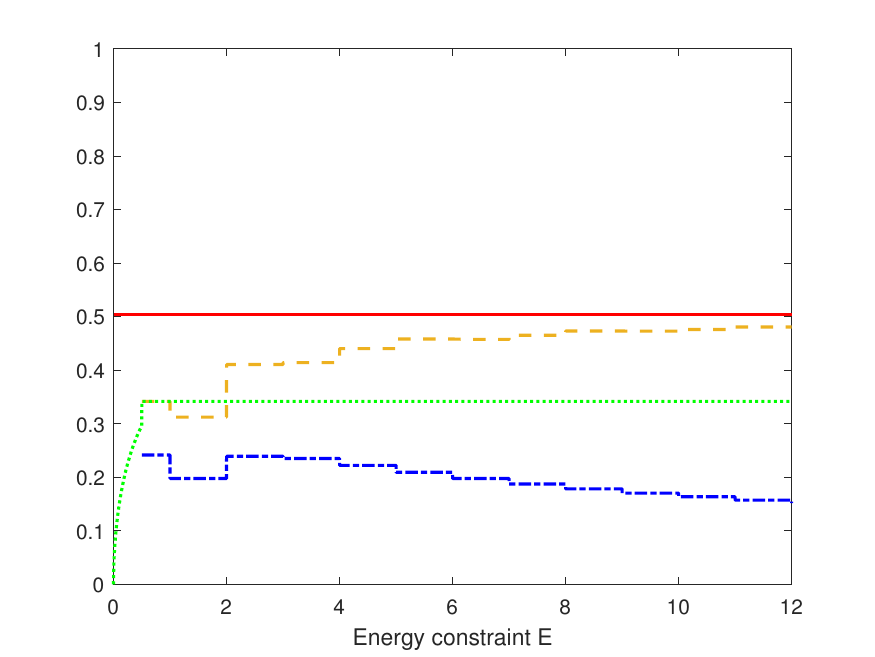}
          \caption{{ For $\gamma_1=0.1$, $\gamma_2=1$ and  $q_1=q_2=0.5$, it is shown vs energy $E$ the quantity $\max_\rho\|\Delta(\rho)\|_1$ of Theorem \ref{thm:theorem2} ({ solid red line}), the bound of Theorem \ref{thm3} ({ dotted green line}),  the bound of Theorem \ref{thm:energyconstriant} ({ dashed gold line}),  and the bound of Corollary \ref{thm:Frobeniusenergyconstriant} ({ dash-dotted blue line}).}}
\label{fig1}
\end{center}
\end{figure}

\bigskip

{ Fig.  \ref{fig2} gives an idea of the tightness of the bound of Theorem \ref{thm:energyconstriant}
for a fixed value of energy. In fact, the figure shows the r.h.s. of Eq.\eqref{thm:energyconstriant} contrasted with the l.h.s. 
evaluated by a maximization over randomly sampled sates satisfying the energy constraint. Clearly the two quantities collapse to zero for $\gamma_1=\gamma_2$.}

\begin{figure}[H]
\begin{center}
          \includegraphics[width=0.55\textwidth]{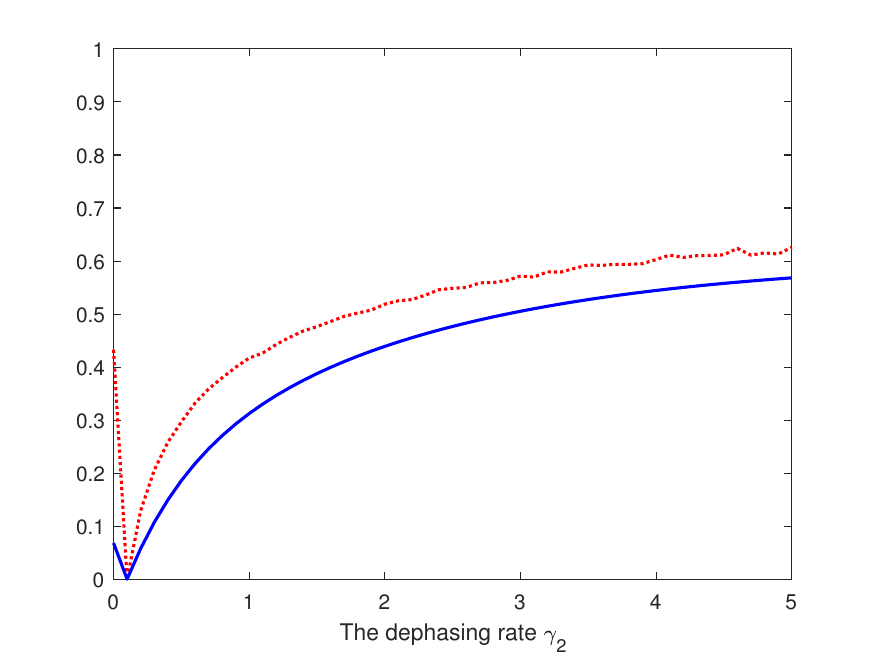}
         \caption{{ For $\gamma_1=0.1$, $q_1=q_2=0.5$ and $E=2$,  it is shown vs $\gamma_2$ the quantity  $\max_{ \rho\in \mathcal{P}(\cH),\, {\rm Tr}(\rho H)\leq E}\|\Delta({ \rho})\|_{1}$ (dotted red line) together with the bound of Theorem \ref{thm:energyconstriant} (solid blue line).}}
         \label{fig2}
\end{center}
\end{figure}


\section{Multiple-shot discrimination}

In this Section we analyze the performance of discrimination across an arbitrary number of trials.
First, let us define $\| \Delta^{(n)}(\rho^{(n)})\|_1:=\|{ q_1} \cN_{\gamma_1}^{\otimes n}(\rho^{(n)})-{ q_2} \cN_{\gamma_2}^{\otimes n}(\rho^{(n)})\|_1$.
Then we have the following.

\begin{theorem}\label{thm:7}
For the discrimination of two bosonic depahsing channels $\cN_{\gamma_1}$ and $\cN_{\gamma_2}$ without energy constraint, we have
\begin{equation}
   \max_{\rho^{(n)}\in {\mathcal{P}}(\cH^{\otimes n})}
   {\| \Delta^{(n)}(\rho^{(n)})\|_1} = \int_{[-\pi, \pi]^n} {\rm d}\theta_1\cdots {\rm d}\theta_n \left\vert ({ q_1}\Pi_{k=1}^n p_{\gamma_1}(\theta_k)-{ q_2}\Pi_{k=1}^n p_{\gamma_2}(\theta_k))\right\vert.
\end{equation}
\end{theorem}

\begin{proof}

Let us first consider a density matrix $\rho^{(2)}\in{\mathcal{P}}(\cH^{\otimes 2})$ with the following representation
\begin{equation}
\rho^{(2)}=\sum_{m_1, n_1, m_2, n_2} \rho_{m_1, n_1, m_2, n_2} \vert m_1\rangle \langle n_1\vert \otimes \vert m_2\rangle\langle n_2\vert.
\end{equation}
Thanks to the continuity of the trace-norm, we can write
\begin{align}
   \cN^{\otimes 2}_{\gamma}(\rho^{(2)})&=\sum_{m_1, n_1, m_2, n_2} \rho_{m_1, n_1, m_2, n_2} \int_{-\pi}^{\pi}\int_{-\pi}^{\pi} {\rm d}\theta_1 {\rm d}\theta_2 p_{\gamma}(\theta_1) p_{\gamma}(\theta_2) e^{-i a^\dagger a\theta_1} \vert m_1\rangle\langle n_1\vert e^{i a^\dagger a\theta_1} \otimes  e^{-i a^\dagger a\theta_2} \vert m_2\rangle\langle n_2\vert e^{i a^\dagger a\theta_2}\\
   &=\sum_{m_1, n_1, m_2, n_2} \rho_{m_1, n_1, m_2, n_2}  \int_{-\pi}^{\pi}\int_{-\pi}^{\pi}  {\rm d}\theta_1 {\rm d}\theta_2 p_{\gamma}(\theta_1) p_{\gamma}(\theta_2) e^{-i\{ (m_1-n_1)\theta_1+ (m_2-n_2)\theta_2\}} \vert m_1\rangle\langle n_1\vert \otimes\vert m_2\rangle\langle n_2\vert \\
   &= \int_{-\pi}^{\pi}\int_{-\pi}^{\pi}  {\rm d}\theta_1 {\rm d}\theta_2 p_{\gamma}(\theta_1) p_{\gamma}(\theta_2) U(\theta_1)^\dagger\otimes U(\theta_2)^\dagger \rho^{(2)}  U(\theta_1)\otimes U(\theta_2),
\end{align}
where $U(\theta)$ is diagonal in the Fock basis with entries $e^{ik\theta}$, $k=0,1,2,\ldots$.

Extending this method to every $\rho^{(n)}\in {\mathcal{P}}(\cH^{\otimes n})$, $n=1, 2, \ldots$, we have
\begin{align}
{\| \Delta^{(n)}(\rho^{(n)})\|_1}
&=\left\Vert\int_{[-\pi, \pi]^n}  {\rm d}\theta_1\cdots {\rm d}\theta_n ({ q_1}\Pi_{k=1}^n p_{\gamma_1}(\theta_k)-{ q_2}\Pi_{k=1}^n p_{\gamma_2}(\theta_k))  \otimes_{k=1}^n U(\theta_k)^\dagger \rho^{(n)}  \otimes_{k=1}^n U(\theta_k) \right\Vert_1\\
&\leq \int_{[-\pi, \pi]^n} {\rm d}\theta_1\cdots {\rm d}\theta_n \left\vert { q_1}\Pi_{k=1}^n p_{\gamma_1}(\theta_k)-{ q_2}\Pi_{k=1}^n p_{\gamma_2}(\theta_k)\right\vert  \left\Vert\otimes_{k=1}^n U(\theta_k)^\dagger \rho^{(n)}  \otimes_{k=1}^n U(\theta_k) \right\Vert_1\\
&=\int_{[-\pi, \pi]^n} {\rm d}\theta_1\cdots {\rm d}\theta_n \left\vert ({ q_1}\Pi_{k=1}^n p_{\gamma_1}(\theta_k)-{ q_2}\Pi_{k=1}^n p_{\gamma_2}(\theta_k))\right\vert.
\end{align}

For the reverse inequality,
following Relation \eqref{re:35}, we define a density matrix $\rho_M^{\otimes n}\in {\mathcal{P}}(\cH^{\otimes n})$ with
\begin{equation}
  \rho_M=\frac{1}{M}T_M.
\end{equation}
It implies
\begin{align}
{\| \Delta^{(n)}(\rho_M^{\otimes n})\|_1}
&=\left\Vert { q_1}\left(T_\infty(p_{\gamma_1}{)}\right)^{\otimes n}\circ \rho_M^{\otimes n}- { q_2}\left(T_\infty(p_{\gamma_2}{)}\right)^{\otimes n}\circ \rho_M^{\otimes n} \right\Vert_1\\
&= \frac{1}{M^n}
\left\Vert P_M^{\otimes n}  \left( { q_1} \left(T_\infty (p_{\gamma_1}{)} \right)^{\otimes n}
- { q_2}\left( T_\infty (p_{\gamma_2}{)} \right)^{\otimes n} \right)     P_M^{\otimes n}
\right\Vert_1,
\end{align}
where $P_M=\sum_{k=0}^{M-1} \vert k\rangle\langle k\vert$ is the projector onto $\cH_M\subseteq \cH$. 

According to Lemma \ref{lemmultilevel} of Appendix \ref{app:multilevel}, $ \mathcal{N}_{\gamma_1}(\rho^{\otimes n})-\mathcal{N}_{\gamma_2}(\rho^{\otimes n}) $ turns out to be a multilevel Toeplitz matrix. Then, taking $F(x)=\vert x\vert$ in Theorem \ref{thm9} of Appendix \ref{app:multilevel}, we arrive at
  \begin{align}
\max_{\rho^{(n)}\in{ \mathcal{P}}(\cH^{\otimes n})} {\| \Delta^{(n)}(\rho^{(n)})\|_1 } 
 &\geq \lim_{k\to \infty}   {\| \Delta^{(n)}(\rho_k^{\otimes n})\|_1} \\
    &=\int_{[-\pi, \pi]^n} {\rm d}\theta_1\cdots {\rm d}\theta_k \left\vert{ q_1} \Pi_{k=1}^n p_{\gamma_1}(\theta_k)-{ q_2}\Pi_{k=1}^n p_{\gamma_2}(\theta_k)\right\vert.\label{re:91}
\end{align}
 \end{proof}

Notice that the result of Theorem \ref{thm:7} was already established in Ref.\cite{PRXQuantum} (see Eq.(12) there), 
however here it is derived with a different approach.

In order to consider the energy constraint, we introduce 
the Hermitian operator 
\begin{equation}
H^{(n)}=\left(H\otimes I\otimes\cdots\otimes I\right)+ \left(I\otimes H\otimes I\otimes \cdots \otimes I\right)+\cdots 
+\left(I\otimes \cdots \otimes I\otimes H\right).
\end{equation}


\begin{theorem}\label{thm:8}
For the discrimination of two bosonic dephasing channels $\cN_{\gamma_1}$ and $\cN_{\gamma_2}$ with energy constraint 
${\rm Tr}(\rho^{(n)} H^{(n)})\leq nE$, $H^{(n)}\in{\cal B}(\cH^{\otimes n})$, $\rho^{(n)}\in {\mathcal{P}}(\cH^{\otimes n})$,  we have the following:
{
\begin{itemize}
  \item If $E\geq 1/2$, then 

\begin{equation}\label{eq:maxalp}
   \max_{\rho^{(n)}\in {\mathcal{P}}(\cH^{\otimes n}),{\rm Tr}(\rho^{(n)} H^{(n)})\leq nE}
   {\| \Delta^{(n)}(\rho^{(n)})\|_1 } \geq 1-
   2\sqrt{q_1 q_2}\left(1- {\alpha_+}  \right)^{\frac{n}{2}},
\end{equation}
with
\begin{equation}
\alpha_+:=\frac{1}{4}\left(\frac{1}{\lfloor 2E\rfloor+1}\right)^2 \left\Vert P_{\lfloor 2E\rfloor+1}\left (T_\infty(p_{\gamma_1})-T_\infty(p_{\gamma_2})\right)\right\Vert_1^2,
\end{equation}
where $P_{\lfloor 2E\rfloor+1}$ is as in Theorem \ref{thm:energyconstriant}.

  \item If $E<1/2$, then 
\begin{equation}
   \max_{\rho^{(n)}\in \mathcal{P}(\cH^{\otimes n}),{\rm Tr}(\rho^{(n)} H^{(n)})\leq nE}
   {\| \Delta^{(n)}(\rho^{(n)})\|_1 } \geq 1-
   2\sqrt{q_1 q_2}\left(1-\alpha_-  \right)^{\frac{n}{2}},
   \end{equation}
with
\begin{equation}\label{re:89}
\alpha_-:=\frac{E}{m_*} \left(1-\frac{E}{m_*}\right)  
   \left\vert e^{-\frac{m_*^2}{2}\gamma_1}-  e^{-\frac{m_*^2}{2}\gamma_2}\right\vert^2.
\end{equation}
where $m_*$ is as in Eq.\eqref{re:46}.
\end{itemize}
}

\end{theorem}

\begin{proof}
{
We exploit the following inequalities
\be
\|A-B\|_1 \geq {\rm{Tr}}(A+B)-2{\rm{Tr}}\left(\sqrt{A} \sqrt{B}\right),
\ee
where $A, B$ are positive operators \cite{ACMB2007}, and
\begin{equation}\label{rhosig}
\Vert \rho-\sigma\Vert_1\leq 2\sqrt{1-\left( \rm{Tr}\sqrt{\rho}\sqrt{\sigma}\right)^2},
\end{equation}
where $\rho,\sigma$ are density operators  \cite{Holevo}, to arrive at
\begin{align}
  \max_{\rho^{(n)}\in { \mathcal{P}}(\cH^{\otimes n}),\, {\rm Tr}(\rho^{(n)} H^{(n)})\leq nE}
  {\| \Delta^{(n)}(\rho^{(n)})\|_1 } 
  &\geq \max_{  \rho\in \mathcal{P}(\cH), \, {\rm Tr}(\rho H)\leq  E}
  {\| \Delta^{(n)}(\rho^{\otimes n})\|_1 } \\
  &\geq  \max_{ \rho\in \mathcal{P}(\cH), \,{\rm Tr}(\rho H)\leq  E} \left\{ 1-2\sqrt{q_1 q_2}\left(\sqrt{1-\frac{ \Vert \cN_{\gamma_1}(\rho)-\cN_{\gamma_2}(\rho) \Vert_1^2 }{4} }   \right)^n \right\}.\label{re:100}
\end{align}

For the case $E\geq 1/2$, 
using Theorem \ref{thm:energyconstriant} for one shot, with $q_1=q_2=1/2$, we have
\begin{equation}\label{re:88}
     \max_{ \rho\in \mathcal{P}(\cH), \,{\rm Tr}(\rho H)\leq  E} \frac{1}{2}\Vert \cN_{\gamma_1}(\rho)-\cN_{\gamma_2}(\rho) \Vert _{1}\geq
     \frac{1}{2(\lfloor 2E\rfloor+1)} \left\Vert P_{\lfloor 2E\rfloor+1}\left (T_\infty(p_{\gamma_1})-T_\infty(p_{\gamma_2})\right)\right\Vert_1,
\end{equation}
where $P_{\lfloor 2E\rfloor+1}=\sum_{n=0}^{\lfloor 2E\rfloor} \vert n\rangle\langle n\vert$.
}
Putting together Eqs.\eqref{re:100} and \eqref{re:88}, we finally obtain
{
\begin{equation}\label{eq:lbntimes}
 { \max_{\rho^{(n)}\in \mathcal{P}(\cH^{\otimes n}), \,{\rm Tr}(\rho^{(n)} H^{(n)})\leq nE}}
 {\| \Delta^{(n)}(\rho^{(n)})\|_1 }  \geq 1-2\sqrt{q_1 q_2}
\left[1-\frac{1}{4}\left(\frac{1}{\lfloor 2E\rfloor+1}\right)^2 \left\Vert P_{\lfloor 2E\rfloor+1}\left (T_\infty(p_{\gamma_1})-T_\infty(p_{\gamma_2})\right)\right\Vert_1^2 \right]^{n/2}.
  \end{equation}

For the case $E<1/2$, using Theorem \ref{thm3} for one shot, with $q_1=q_2=1/2$, we have
\begin{equation}\label{re:101}
     \max_{ \rho\in \mathcal{P}(\cH_2), \,{\rm Tr}(\rho H)\leq  E}\frac{1}{2}\Vert \cN_{\gamma_1}(\rho)-\cN_{\gamma_2}(\rho) \Vert _{1}\geq
     \sqrt{\frac{E}{m_*} \left(1-\frac{E}{m_*}\right)} \,  
    \left\vert  e^{-\frac{m_*^2}{2}\gamma_1}-  e^{-\frac{m_*^2}{2}\gamma_2}\right\vert.
\end{equation}
}
Putting together Eqs.\eqref{re:100} and \eqref{re:101} we finally obtain
\begin{equation}\label{eq:lbntimes}
{ \max_{\rho^{(n)}\in \mathcal{P}(\cH^{\otimes n}),\,{\rm Tr}(\rho^{(n)} H^{(n)})\leq nE}
 {\| \Delta^{(n)}(\rho^{(n)})\|_1 } \geq  
  1- 2\sqrt{q_1 q_2}\left[1-\frac{E}{m_*} \left(1-\frac{E}{m_*}\right)  
   \left\vert e^{-\frac{m_*^2}{2}\gamma_1}-  e^{-\frac{m_*^2}{2}\gamma_2}\right\vert^2
 \right]^{n/2}.
 }
   \end{equation}
\end{proof} 

{ In Fig.\ref{fig3} we can see that the bound of Theorem \ref{thm:8} saturates to a constant value. However, this latter is well separated  from the value of Theorem \ref{thm:7}, although for increasing $n$ the gap tends to reduce (see inset).
This indicates that the bound of Theorem \ref{thm:8} is not much tight. A fact that is shown by Fig.\ref{fig4}, where 
the r.h.s. of Eq.\eqref{eq:maxalp} is contrasted with the l.h.s. evaluated by a maximization over randomly sampled sates satisfying the
energy constraint. Clearly the two quantities collapse to zero for $\gamma_1 = \gamma_2$.}
\medskip

\begin{figure}[H]
\centering
\begin{overpic}[scale=0.8]{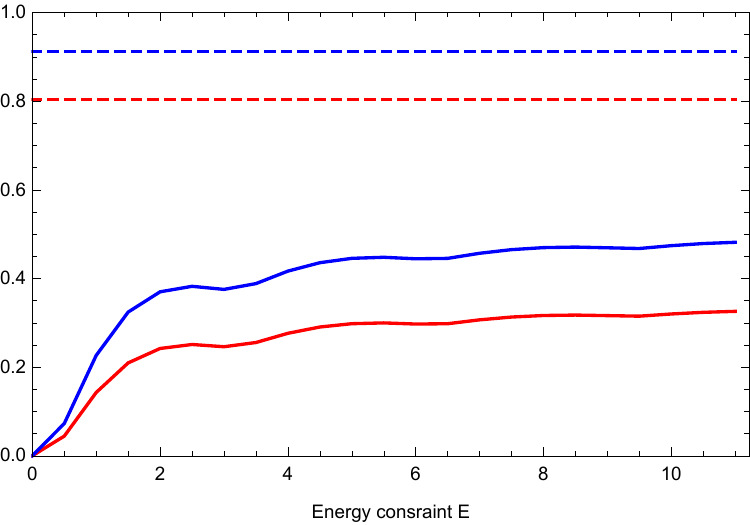}
 \put(64,43){\color{black}%
 \frame{\includegraphics[scale=.3]{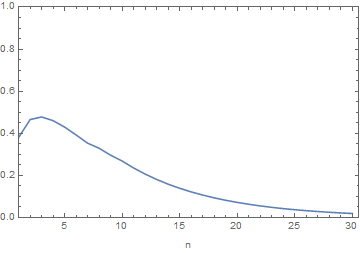}}}
 \end{overpic}
 \caption{{ For $\gamma_1=0.1$, $\gamma_2=1$ and $q_1=q_2=0.5$, it is shown vs energy $E$ the quantity $\max_{\rho^{(n)}\in \mathcal{P}(\cH^{\otimes n})}
   {\| \Delta^{(n)}(\rho^{(n)})\|_1}$ of Theorem \ref{thm:7} (dotted lines) and the 
           lower bound of Theorem \ref{thm:8} (solid lines). The colors red and blue refer respectively to $n=3, 5$. The inset shows the difference between 
the quantity $\max_{\rho^{(n)}\in \mathcal{P}(\cH^{\otimes n})}
   {\| \Delta^{(n)}(\rho^{(n)})\|_1}$ of Theorem \ref{thm:7} and the 
           lower bound of Theorem \ref{thm:8} for $E=12$ vs $n$.}}
\label{fig3}
\end{figure}

\begin{figure}[H]
\begin{center}
          \includegraphics[width=0.55\textwidth]{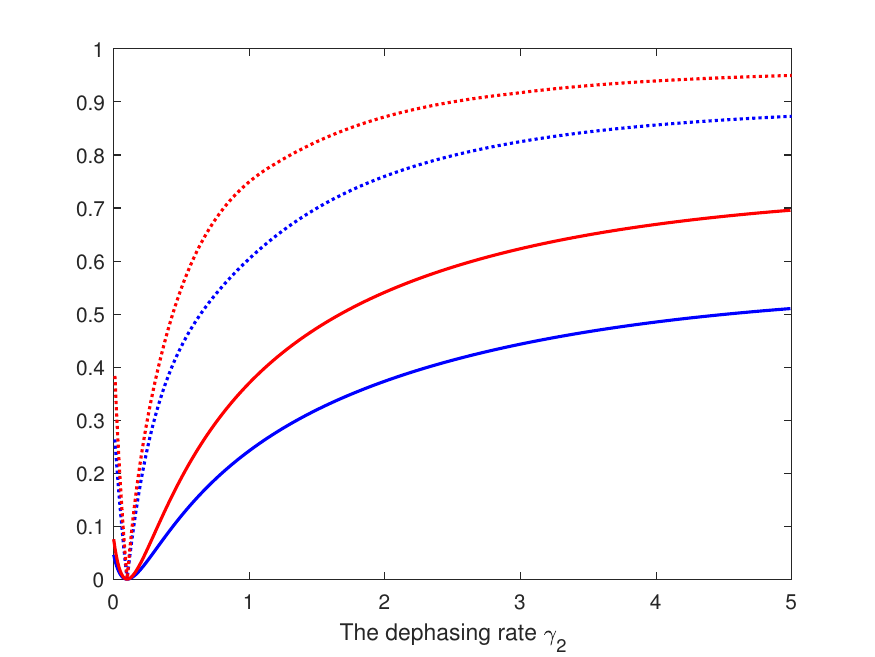}
          \caption{{ For $\gamma_1=0.1$, $q_1=q_2=0.5$ and $E=2$, it is shown vs $\gamma_2$  the quantity $\max_{\rho^{(n)}\in \mathcal{P}(\cH^{\otimes n}),
          \, {{\rm Tr}(\rho^{(n)}H^{(n)}\leq nE }}
   {\| \Delta^{(n)}(\rho^{(n)})\|_1}$ (dotted lines) together with the 
          lower bound of Theorem \ref{thm:8} (solid lines). The colors blue and red refer respectively to $n=3,5$.}}
\label{fig4}
\end{center}
\end{figure}

\medskip

{
From Theorem \ref{thm:8} and relation \eqref{perr} we can argue 
\begin{equation}\label{perrbound}
  p_{err}\leq \frac{1}{2}\exp\left[n\ln \sqrt{(1-{\alpha_{\pm}})} \right].
\end{equation}
In Eq.\eqref{perrbound} the quantity $\ln\sqrt{1-{\alpha_{\pm}}}$ represents the exponential decay rate of the upper bound to the error probability.

Notice that exploiting the following chain of inequalities \cite{ACMB2007}
 \begin{equation}\label{rhosigQ}
  1- Q(\rho, \sigma)\leq \frac{1}{2} \Vert \rho -\sigma\Vert_1 \leq \sqrt{1-Q^{2}(\rho, \sigma)} ,
 \end{equation}
 where
 \begin{equation}
   Q(\rho, \sigma):=\inf_{s\in [0, 1]} {\rm Tr}(\rho^s \sigma^{1-s}),
 \end{equation}
relation \eqref{rhosigQ} leads to the quantum Chernoff bound for discriminating between $\rho$ and $\sigma$ on $n$ trials, namely
 \begin{equation}
   p_{err}\left(\rho^{\otimes n},\sigma^{\otimes n}\right)\leq \frac{1}{2}\exp\left[ n\ln Q(\rho,\sigma)\right].
 \end{equation}}
Then, similarly to \eqref{perrbound}, we could write
\begin{align}\label{re:93}
p_{err}&\equiv
\min_{\rho^{(n)}\in {\mathcal{P}}(\cH^{\otimes n}),{\rm Tr}(\rho^{(n)} H^{(n)})\leq nE} p_{err}
\left({\cal N}^{\otimes n}_{\gamma_1}(\rho^{(n)}),{\cal N}^{\otimes n}_{\gamma_2}(\rho^{(n)})\right)\notag \\
&\leq 
\min_{\rho\in {\mathcal{P}}(\cH), {\rm rank}{\rho}=2,{\rm Tr}(\rho H)\leq  E}
\frac{1}{2}\exp\left[ n\ln Q\left({ q_1}{\cal N}_{\gamma_1}(\rho),{ q_2}{\cal N}_{\gamma_2}(\rho)\right)\right].
\end{align}
However, the calculation of $Q\left({ q_1}{\cal N}_{\gamma_1}(\rho),{ q_2}{\cal N}_{\gamma_2}(\rho)\right)$ is challenging. Moreover, it is not known 
if \eqref{re:93} results tighter than \eqref{perrbound}.


\section{Conclusion}

We studied the possibility of discriminating between two bosonic dephasing quantum channels.
We demonstrated that unambiguous discrimination is not feasible. Subsequently, we considered discrimination with nonzero error probability and analytically found the minimum over all possible input states, providing an alternative derivation with respect to Ref.\cite{PRXQuantum}.
{ Our proofs seem more straightforward than that in Ref.\cite{PRXQuantum}. In fact, for the optimality
part, they simply rely on basic properties of the trace norm, while Ref.\cite{PRXQuantum}
resorts to the possibility of simulating the channel by means of adjoining a parametrized
environment state following by the action of an un-parametrized channel. Additionally,
for the achievability part, given that the action of a bosonic dephasing
quantum channel can be described in terms of Toeplitz matrix (see Eq.\eqref{re:14}), it
seems quite natural to exploit properties of these matrices. Thus, we used the well
known Theorem \ref{thm:sezgo} and its extension to multilevel Toeplitz matrices (Theorem \ref{thm9}
in Appendix \ref{app:multilevel}), instead of studying the convergence of a sequence of probability
distributions obtained after applying at the output of the channel suitable quantum
and classical post processes. This as well gives insights on the optimal input state
(see the form of Eq.\eqref{re:35} that is analogous to the case discussed in Sec.VI A of
Ref.\cite{PRXQuantum}.)
}

Then, we went beyond the unconstrained input states and we derived an upper bound on the error probability under input energy constraint.
These results are extended from single-shot to multi-shot discrimination, considering the asymptotic behavior. 
The problem of establishing an order relation between the derived bound and the Chernoff bound remains open.
Additionally, the analysis performed for multi-shot discrimination considered non-local measurements. Hence, it would desirable to also determine the achievable performance with local measurements, eventually supplied with an adaptive strategy.    

{Finally, entanglement is often a useful resource for enhancing the channel discrimination capabilities \cite{PECD}. This seems not the case for dephasing quantum channels where, in absence of input constraints, the inutility of side entanglement is shown in Appendix \ref{app:side},
paralleling the result found in Ref.\cite{DisDeph} for finite dimensional dephasing quantum channels. However, whether or not entanglement assistance can provide any advantage in the presence of energy constraint is left for future investigations.}


\subsection*{Acknowledgments}
S.Kh.O. is grateful to Carlo Garoni from Department of Mathematics, University of Rome Tor Vergata, Italy, 
for useful discussions about multilevel Toeplitz matrices.
L.M. acknowledges financial support from the Iran National Science Foundation (INSF) under Project No. 4022322.
M.R. and S.M. acknowledge financial support from ``PNRR MUR project PE0000023-NQSTI''.
{S.K.O. and S.M. are also grateful to M.M. Wilde for useful comments.}


\bibliographystyle{IEEEtran}
\bibliography{bib}

\begin{thebibliography}{10}
\providecommand{\url}[1]{#1}
\csname url@samestyle\endcsname
\providecommand{\newblock}{\relax}
\providecommand{\bibinfo}[2]{#2}
\providecommand{\BIBentrySTDinterwordspacing}{\spaceskip=0pt\relax}
\providecommand{\BIBentryALTinterwordstretchfactor}{4}
\providecommand{\BIBentryALTinterwordspacing}{\spaceskip=\fontdimen2\font plus
\BIBentryALTinterwordstretchfactor\fontdimen3\font minus
  \fontdimen4\font\relax}
\providecommand{\BIBforeignlanguage}[2]{{%
\expandafter\ifx\csname l@#1\endcsname\relax
\typeout{** WARNING: IEEEtran.bst: No hyphenation pattern has been}%
\typeout{** loaded for the language `#1'. Using the pattern for}%
\typeout{** the default language instead.}%
\else
\language=\csname l@#1\endcsname
\fi
#2}}
\providecommand{\BIBdecl}{\relax}
\BIBdecl

\bibitem{GLN05}
A.~Gilchrist, N.~K. Langford, and M.~A. Nielsen, ``Distance measures to compare
  real and ideal quantum processes,'' \emph{Phys. Rev. A}, vol.~71, p. 062310,
  Jun 2005.

\bibitem{S05}
M.~F. Sacchi, ``Optimal discrimination of quantum operations,'' \emph{Phys.
  Rev. A}, vol.~71, p. 062340, Jun 2005.

\bibitem{WY06}
G.~Wang and M.~Ying, ``Unambiguous discrimination among quantum operations,''
  \emph{Phys. Rev. A}, vol.~73, p. 042301, Apr 2006.

\bibitem{H09}
M.~Hayashi, ``Discrimination of two channels by adaptive methods and its
  application to quantum system,'' \emph{IEEE Transactions on Information
  Theory}, vol.~55, no.~8, pp. 3807--3820, 2009.

\bibitem{MrSm}
M.~Rexiti and S.~Mancini, ``Discriminating qubit amplitude damping channels,''
  \emph{Journal of Physics A: Mathematical and Theoretical}, vol.~54, no.~16,
  p. 165303, mar 2021.

\bibitem{DFY09}
R.~Duan, Y.~Feng, and M.~Ying, ``Perfect distinguishability of quantum
  operations,'' \emph{Phys. Rev. Lett.}, vol. 103, p. 210501, Nov 2009.

\bibitem{AMM20}
A.~Arqand, L.~Memarzadeh, and S.~Mancini, ``Quantum capacity of a bosonic
  dephasing channel,'' \emph{Physical Review A}, vol. 102, p. 042413, 2020.

\bibitem{Dehdashti2022QuantumCO}
J.~N. Shahram~Dehdashti and P.~van Loock, ``Quantum capacity of a deformed
  bosonic dephasing channel,'' \emph{arXiv:2211.09012}, 2022.

\bibitem{LW23}
L.~Lami and M.~M. Wilde, ``Exact solution for the quantum and private
  capacities of bosonic dephasing channels,'' \emph{Nature Photonics}, vol.~17,
  pp. 525--530, 2022.

\bibitem{AMM23}
A.~Arqand, L.~Memarzadeh, and S.~Mancini, ``Energy-constrained locc-assisted
  quantum capacity of the bosonic dephasing channel,'' \emph{Entropy}, vol.~25,
  no. 1001, 2023.

\bibitem{PRXQuantum}
Z.~Huang, Lami, Ludovico, and M.~M. Wilde, ``Exact quantum sensing limits for
  bosonic dephasing channels,'' \emph{PRX Quantum}, vol.~5, p. 020354, Jun
  2024.

\bibitem{HJ91}
R.~A. Horn and C.~R. Johnson, \emph{Topics in Matrix Analysis}.\hskip 1em plus
  0.5em minus 0.4em\relax Cambridge University Press, 1991.

\bibitem{MW20}
S.~Mancini and A.~Winter, \emph{A Quantum Leap in Information Theory}.\hskip
  1em plus 0.5em minus 0.4em\relax World Scientific Publishing Company Pte
  Limited, 2019.

\bibitem{Sz20}
G.~Szegő, ``Beitr{\"a}ge zur theorie der toeplitzschen formen,''
  \emph{Mathematische Zeitschrift}, vol.~6, pp. 167--202, 1920.

\bibitem{GS84}
U.~Grenander and G.~Szego, \emph{Toeplitz Forms and Their Applications}.\hskip
  1em plus 0.5em minus 0.4em\relax Berkeley: University of California Press,
  1958.

\bibitem{ACMB2007}
K.~M.~R. Audenaert, J.~Calsamiglia, R.~Mu\~noz Tapia, E.~Bagan, L.~Masanes,
  A.~Acin, and F.~Verstraete, ``Discriminating states: The quantum chernoff
  bound,'' \emph{Phys. Rev. Lett.}, vol.~98, p. 160501, Apr 2007.

\bibitem{Holevo}
A.~S. Kholevo, ``On quasiequivalence of locally normal states,''
  \emph{Theoretical and Mathematical Physics}, vol.~13, pp. 1071--1082,
  November 1972.

\bibitem{PECD}
S.~K. Oskouei, S.~Mancini, and M.~Rexiti, ``Profitable entanglement for channel
  discrimination,'' \emph{Proceedings of the Royal Society A}, vol. 479, no.
  20220796, 2023.

\bibitem{DisDeph}
M.~Rexiti, L.~Memarzadeh, and S.~Mancini, ``Discrimination of dephasing
  channels,'' \emph{Journal of Physics A: Mathematical and Theoretical},
  vol.~55, no.~24, p. 245301, may 2022.

\bibitem{TZ98}
E.~E. Tyrtyshnikov, ``A unifying approach to some old and new theorems on
  distribution and clustering,'' \emph{Linear Algebra and its Applications},
  vol. 232, pp. 1--43, 1996.

\bibitem{S99}
S.~Serra, ``Spectral and computational analysis of block toeplitz matrices
  having nonnegative definite matrix-valued generating functions,'' \emph{BIT
  Numerical Mathematics}, vol.~39, no.~1, 1999.

\bibitem{GS17}
C.~Garoni and S.~Serra-Capizzano, \emph{Generalized locally Toeplitz sequences:
  theory and applications}.\hskip 1em plus 0.5em minus 0.4em\relax Berlin:
  Springer Publishing, 2017, vol.~1.

\end{thebibliography}

\appendix
\section{Multilevel Toeplitz matrices}\label{app:multilevel}

In this Appendix, we recall basic features of multilevel Toeplitz matrices that are useful for our purposes
\cite{TZ98, S99, GS17}. We start noticing that for a given function $f:[-\pi, \pi]^n\to \mathbb{C}$, with $f\in L^1([-\pi, \pi]^n)$, the multidimensional Fourier series of $f$ is given by
\begin{equation}
  f(\theta_1, \ldots, \theta_n)=\sum_{k_1, \cdots, k_n\in\mathbb{Z}} f_{k_1, \cdots, k_n} e^{i(k_1\theta_1+\ldots+ k_n \theta_n)},
\end{equation}
where $f_{k_1, \ldots, k_n}$ are the Fourier coefficients of $f$, namely
\begin{equation}
  f_{k_1, \ldots, k_n}=\frac{1}{(2\pi)^n}\int_{[-\pi, \pi]^n} {\rm d}\theta_1\cdots {\rm d}\theta_n f(\theta_1, \cdots, \theta_n) e^{-i(k_1\theta_1+\cdots, k_n\theta_n)}.
\end{equation}
The symmetric multilevel Toeplitz matrix $T^{(n)}_{M_1\cdots M_n}(f)$ 
 of size $M_1M_2 \cdots M_n$, for every $M_1, M_2,  \ldots, M_n, n\in \mathbb{N}$, is defined by induction over $n$ in the following manner.

Using Relation \eqref{defTmatrix}, the $1$-level Toeplitz matrices of size $M_1$ is denoted as $T_{M_1}^{(1)}$. 
The symmetric $n$-level Toeplitz matrix of size $M_1\cdots M_n$ is defined based on $(n-1)$-level Toeplitz matrices as follows 
\begin{equation}
T^{(n)}_{M_1\cdots M_n}(f)=  \left(
    \begin{array}{cccc}
      T_0^{(n-1)} & T_{1}^{(n-1)} & \cdots & T_{M_{n-1}-1}^{(n-1)} \\
      T_{1}^{(n-1)} & \ddots & \ddots & \vdots \\
      \vdots & \ddots & \ddots & T^{(n-1)}_{1} \\
      T_{M_{n-1}-1}^{(n-1)} & \cdots & T_{1}^{(n-1)} & T_0^{(n-1)} \\
    \end{array}
  \right),
\end{equation}
where each block $T_{k}^{(n-1)}$, $k= 0, \cdots, M_{n-1}-1$ is itself a $(n-1)$-level Toeplitz matrix of size $M_1 M_2\cdots M_{n-1}$. In other words, we have \cite{TZ98}
\begin{equation}
     [T^{(n)}_{M_1M_2\cdots M_n}(f)]=\left[ \cdots \left[[f_{|i_1-j_1|, |i_2-j_2|, |i_3-j_3|, \cdots, |i_n-j_n|}]_{i_1, j_1=0}^{M_1-1}\right]_{i_2, j_2=0}^{M_2-1}\cdots \right]_{i_n, j_n=1}^{M_n-1}.
\end{equation} 

To clarify the structure of multilevel Toeplitz matrices, let us explicitly show the $3$-level Toeplitz matrix of size $8$ generated by a function $f: [-\pi, \pi]^3\to \mathbb{C}$
\begin{equation}
  T_{2\times2\times 2}^{(3)}(f)=\left(
    \begin{array}{cc|cc||cc|cc}
      f_{0, 0, 0} & f_{0, 0, 1} & f_{0, 1, 0} & f_{0, 1, 1} & f_{1, 0, 0} & f_{1, 0, 1} & f_{1, 1, 0} & f_{1, 1, 1} \\
      f_{0, 0, 1} & f_{0, 0, 0}  &f_{0, 1, 1} & f_{0, 1, 0} &  f_{1, 0, 1} & f_{1, 0, 0}  &f_{1, 1, 1} & f_{1, 1, 0}\\
      \hline
         f_{0, 1, 0} & f_{0, 1, 1} & f_{0, 0, 0} & f_{0, 0, 1}&  f_{1, 1, 0} & f_{1, 1, 1} & f_{1, 0, 0} & f_{1, 0, 1} \\
      f_{0, 1, 1} & f_{0, 1, 0}&f_{0, 0, 1} & f_{0, 0, 0}& f_{1, 1, 1} & f_{1, 1, 0}&f_{1, 0, 1} & f_{1, 0, 0}\\
      \hline\hline
       f_{1, 0, 0} & f_{1, 0, 1} & f_{1, 1, 0} & f_{1, 1, 1} & f_{0, 0, 0} & f_{0, 0, 1} & f_{0, 1, 0} & f_{0, 1, 1} \\
      f_{1, 0, 1} & f_{1, 0, 0}  &f_{1, 1, 1} & f_{1, 1, 0} &  f_{0, 0, 1} & f_{0, 0, 0}  &f_{0, 1, 1} & f_{0, 1, 0}\\
      \hline
         f_{1, 1, 0} & f_{1, 1, 1} & f_{1, 0, 0} & f_{1, 0, 1}&  f_{0, 1, 0} & f_{0, 1, 1} & f_{0, 0, 0} & f_{0, 0, 1} \\
      f_{1, 1, 1} & f_{1, 1, 0}&f_{1, 0, 1} & f_{1, 0, 0}& f_{0, 1, 1} & f_{0, 1, 0}&f_{0, 0, 1} & f_{0, 0, 0}\\
    \end{array}
  \right).
\end{equation} 

\begin{lemma}\label{lemmultilevel}
Given two symmetric  Toeplitz matrices $T_M(t^{(1)})$ and $T_M(t^{(2)})$, $(T_M(t^{(1)}))^{\otimes n}-(T_M(t^{(2)}))^{\otimes n}$ is a $n$-level Toeplitz matrix of size $M^n$ which is generated by $t^{(1)}_{k_1}\ldots t^{(1)}_{k_{n}}-t^{(2)}_{k_1}\ldots t^{(2)}_{k_{n}}$, where $k_1, \ldots, k_n
\in \{0,1,\ldots, M-1\}$.
\end{lemma}

\begin{proof}
 We define $J_M^{(k)}$ of size  $M$ as a matrix with entries
\begin{equation}
  \left(  J_M^{(k)} \right)_{ij}=\delta_{|i-j|-k},\quad i, j=1, \cdots, M, \quad k=0, 1, \cdots, n-1,
\end{equation}
where $\delta_a:=1$ if $a=0$ and $\delta_a:=0$ otherwise. 
We can expand $T_{M}(t)$ on the $ J_M^{(k)}$ elements as follows
\begin{equation}
  T_{M}(t^{(j)})=\sum_{k=0}^{M-1} t^{(j)}_k  J_M^{(k)},\qquad k=0, \ldots, M-1,\quad j=1, 2,
\end{equation}
where $t^{(j)}_k$ are the Fourier coefficients of the function $t^{(j)}$. Now, we have
\begin{align}
  \left(T_{M}(t^{(1)})\right)^{\otimes n}-\left(T_{M}(t^{(2)})\right)^{\otimes n} &= \left( \sum_{k=0}^{M-1} t^{(1)}_k  J_M^{(k)}   \right)^{\otimes n}-\left( \sum_{k=0}^{M-1} t^{(2)}_k  J_M^{(k)}   \right)^{\otimes n}\\
   &=\sum_{k_1, \ldots, k_{n}=0}^{M-1} \left(t^{(1)}_{k_1}\cdots t^{(1)}_{k_{n}}-t^{(2)}_{k_1}\cdots t^{(2)}_{k_{n}}\right) J_{M}^{(k_1)}\otimes J_{M}^{(k_2)}\otimes \cdots \otimes J_{M}^{(k_{n})}
\end{align}
Using  \cite[Lemma 3.3]{GS17}, the generating function of $n$-level Toeplitz matrix $ \left(T_{M}(t^{(1)})\right)^{\otimes n}-\left(T_{M}(t^{(2)})\right)^{\otimes n}$ is given by Fourier coefficients $f_{k_1, \cdots, k_n}=t^{(1)}_{k_1}\cdots t^{(1)}_{k_{n}}-t^{(2)}_{k_1}\cdots t^{(2)}_{k_{n}}$.
\end{proof}
Finally, we state the extension of Theorem \ref{thm:sezgo} to multilevel Toeplitz matrices \cite{S99, GS17, TZ98}.
\begin{theorem}\label{thm9}
 Let $f:[-\pi, \pi]^{n}\to \mathbb{C}$ be an absolutely integrable
	function with $\inf_{x\in [-\pi, \pi]^{n} }f(x)$,  $\sup_{x\in [-\pi, \pi]^{n} }f(x)$  finite numbers.  If $T^{(n)}_{M^n}(f)$ is a $n$-level Toeplitz matrix of size $ M^n$ and $\lambda_j(T_{M^n}^{(n)}(f))$ denotes the $j^{th}$ eigenvalue of $T_{M^n}^{(n)}(f)$, then 
	\begin{equation}
		\lim_{M\to\infty} \frac{1}{M^n}\sum_{j} F(\lambda_j(T_{M^n}^{(n)}(f)))=\frac{1}{(2\pi)^n}\int_{[-\pi, \pi]^n} {\rm d}\theta_1\cdots {\rm d}\theta_n F(f(\theta_1, \cdots, \theta_n)),
	\end{equation}
	where $F:[\inf_{x\in [-\pi, \pi]^{n} }f(x), \sup_{x\in [-\pi, \pi]^{n} }f(x)]\to  \mathbb{R}$ is any continuous function.
\end{theorem}


\section{Discrimination with side entanglement}\label{app:side}

 In case we are allowed to use a two-mode entangled state to probe the channel (acting e.g. only on the second mode), the optimization problem \eqref{re:31} becomes $\max_{\rho\in \mathcal{P}(\cH\otimes \cH)}  \| \left({\rm id}\otimes \Delta\right)(\rho)\|_{1}$.
 
\begin{theorem}\label{thm:side}
Without input energy constrain, it is
  \begin{equation}
\max_{\rho\in \mathcal{P}(\cH\otimes \cH)}  \| \left({\rm id}\otimes \Delta\right)(\rho)\|_{1}
\leq 
\max_{\rho\in \mathcal{P}(\cH)}  \|\Delta(\rho)\|_{1}.
\end{equation}
\end{theorem}

\begin{proof}
Following the first steps of the proof of Theorem \ref{thm:theorem2} we can get
\begin{align}
\max_{\rho\in \mathcal{P}(\cH\otimes \cH)}  \| \left({\rm id}\otimes \Delta\right)(\rho)\|_{1}& =  \left\Vert   \int_{-\pi}^{\pi} {\rm d}\theta \left({ q_1 } p_{\gamma_1}(\theta)-{ q_2}p_{\gamma_2}(\theta)\right) \sum_{m_1,n_1, m_2, n_2=0}^\infty
e^{i\theta(n_2-m_2)}  \rho_{m_1,n_1, m_2, n_2}\vert m_1\rangle\langle  n_1\vert\otimes \vert m_2\rangle\langle  n_2\vert \right\Vert_1.
\end{align}
Using $U(\theta)$ defined as a diagonal matrix in the Fock basis with entries  $e^{ik\theta}$, $k=0,1,2,\cdots$, we have
\begin{align}
\max_{\rho\in \mathcal{P}(\cH\otimes \cH)}  \| \left({\rm id}\otimes \Delta\right)(\rho)\|_{1} &= \left\Vert  \int_{-\pi}^{\pi} {\rm d}\theta \left({ q_1}p_{\gamma_1}(\theta)-{ q_2}p_{\gamma_2}(\theta)\right) (I \otimes U(\theta)^\dagger)\rho (I \otimes U(\theta)) \right\Vert_1\\
&\leq \int_{-\pi}^{\pi} {\rm d}\theta \left\vert\left({ q_1}p_{\gamma_1}(\theta)-{ q_2}p_{\gamma_2}(\theta)\right)\right\vert \left\Vert (I \otimes U(\theta)^\dagger)\rho 
(I \otimes U(\theta))\right\Vert_1\\
&=   \int_{-\pi}^{\pi} {\rm d}\theta \left\vert { q_1}p_{\gamma_1}(\theta)-{ q_2}p_{\gamma_2}(\theta)\right\vert\\
&=\max_{\rho\in \mathcal{P}(\cH)}  \|\Delta(\rho)\|_{1},
\end{align}
where in the last equality we used the result of Theorem \ref{thm:theorem2}.
\end{proof}
Since by definition it is
\begin{equation}
\max_{\rho\in \mathcal{P}(\cH)}  \|\Delta(\rho)\|_{1}\leq  \max_{\rho\in \mathcal{P}(\cH\otimes \cH)}  \| \left({\rm id}\otimes \Delta\right)(\rho)\|_{1},
\end{equation}
together with Theorem \ref{thm:side} we can conclude that 
$\max_{\rho\in \mathcal{P}(\cH)}  \|\Delta(\rho)\|_{1} =  \max_{\rho\in \mathcal{P}(\cH\otimes \cH)}  \| \left({\rm id}\otimes \Delta\right)(\rho)\|_{1}$,
i.e. that side entanglement is not useful in discrimination between two bosonic dephasing quantum channels.


\end{document}